\documentclass[12pt,reqno]{amsart}
\usepackage{amsmath}
\usepackage{latexsym}
\usepackage{amsfonts}
\usepackage{amssymb}
\usepackage{color}
\usepackage{bbm,dsfont}
\usepackage{graphicx}
\usepackage{hyperref}

\newtheorem{proposition}{Proposition}
\newtheorem{proposition?}{Proposition?}
\newtheorem{theorem}{Theorem}
\newtheorem{lemma}{Lemma}

\theoremstyle{definition}

\newtheorem{example}{Example}
\newtheorem{definition}{Definition}

\newcommand{\complex}{\mathbb C} 
\newcommand{\nat}{\mathbb N} 
\newcommand{\half}{\tfrac{1}{2}} 

\newcommand{\hi}{\mathcal{H}} 
\newcommand{\hik}{\mathcal{K}} 
\newcommand{\hv}{\mathcal{V}} 
\newcommand{\lh}{\mathcal{L(H)}} 
\newcommand{\lk}{\mathcal{L(K)}} 
\newcommand{\trh}{\mathcal{T(H)}} 
\newcommand{\trk}{\mathcal{T(K)}} 
\newcommand{\sh}{\mathcal{S(H)}} 
\newcommand{\ip}[2]{\left\langle\,#1\,|\,#2\,\right\rangle} 
\newcommand{\kb}[2]{|#1\rangle\langle#2|} 
\newcommand{\no}[1]{\left\|#1\right\|} 
\newcommand{\tr}[1]{\textrm{tr}\left[#1\right]} 
\newcommand{\trvin}[1]{\textrm{tr}_{\hv_1}[#1]} 
\newcommand{\trvout}[1]{\textrm{tr}_{\hv_2}[#1]} 
\newcommand{\id}{\mathbbm{1}} 

\newcommand{\dev}{\mathfrak{D}}

\newcommand{\A}{\mathsf{A}}
\newcommand{\E}{\mathsf{E}}
\newcommand{\F}{\mathsf{F}}
\newcommand{\G}{\mathsf{G}}

\newcommand{\I}{\mathcal{I}}

\newcommand{\memo}{\mathfrak{M}}
\newcommand{\swap}{U_{\mathsf{swap}}}

\newcommand{\hh}{^H} 
\renewcommand{\ss}{^S} 

\begin{document}

\title{Strongly Incompatible Quantum Devices}

\author[Heinosaari]{Teiko Heinosaari$^\natural$}
\address{$\natural$ Turku Centre for Quantum Physics, Department of Physics and Astronomy, University of Turku}
\email{teiko.heinosaari@utu.fi}

\author[Miyadera]{Takayuki Miyadera$^\flat$}
\address{$\flat$ Department of Nuclear Engineering, Kyoto University, 6158540 Kyoto, Japan}
\email{miyadera@nucleng.kyoto-u.ac.jp}

\author[Reitzner]{Daniel Reitzner$^\clubsuit$}
\address{$\clubsuit$ Department of Mathematics, Technische Universit\"at M\"unchen, 85748 Garching, Germany
\newline \indent and \newline
 Research Center for Quantum Information, Slovak Academy of Sciences, D\'ubravsk\'a cesta 9, 845 11 Bratislava, Slovakia}
\email{daniel.reitzner@savba.sk}

\date{\today}

\begin{abstract}
The fact that there are quantum observables without a simultaneous measurement is one of the fundamental characteristics of quantum mechanics.
In this work we expand the concept of joint measurability to all kinds of possible measurement devices, and we call this relation \emph{compatibility}. 
Two devices are \emph{incompatible} if they cannot be implemented as parts of a single measurement setup.
We introduce also a more stringent notion of incompatibility, \emph{strong incompatibility}. 
Both incompatibility and strong incompatibility are rigorously characterized and their difference is demonstrated by examples. 
\end{abstract}


\maketitle

\section{Introduction}\label{sec:intro}

Incompatibility of two quantum observables means that they cannot be implemented in a single measurement setup.  
The existence of incompatible observables is a genuine quantum phenomenon, and it is perhaps most notably manifested in various uncertainty relations. 
The best known examples of incompatible observables are the spin components in orthogonal directions and the canonical pair of position and momentum.

Most of the earlier studies on incompatibility have been concentrating on observables and effects (see e.g.~\cite{Lahti03} for a survey).
In this work we define the notion of incompatibility in a general way so that it becomes possible to speak about incompatibility of two different types of devices, e.g.~incompatibility of an observable and a channel, or an effect and an operation.
Our proposed definition is a straightforward generalization of the usual one for observables; two devices are incompatible if they cannot be parts of a single measurement setup.  

Our approach provides the possibility to separate two qualitatively different levels of incompatibility.
Namely, we will define the concept of strong incompatibility and demonstrate that this is, indeed, more stringent condition than mere incompatibility.
Strongly incompatible devices cannot be implemented on the same measurement setup even if we were allowed to change one of its specific parts, the pointer observable.

It is illustrative to note some similarities between (strong) incompatibility and entanglement.
Typically, entanglement is defined through its negation -- separability.
In a similar way, incompatibility is defined through its negation -- compatibility.
It is easy to intuitively grasp the notions of separable states and compatible measurements, while entangled states and incompatible measurements are harder to comprehend.
One of the best indication of a quantum regime, namely a violation of a Bell inequality, requires both an entangled state and a collection of incompatible observables.

Our investigation is organized as follows.
In Section \ref{sec:incompatible-devices} we define incompatibility and strong incompatibility using the property ``being part of an instrument'' as a starting point.
In Section \ref{sec:formulations} we show that it is possible to formulate these relations in terms of Stinespring and Kraus representations.
The connection of incompatibility and strong incompatibility to measurement models is explained in Section \ref{sec:memo}, and this clarifies the operational meaning of the relations.
In Section \ref{sec:examples} we demonstrate all possible relations between operations and effects.
In particular, it will become clear that strong incompatibility is a stricter relation than mere incompatibility.

To this end, let us fix the notation.
Let $\hi$ be either finite or countably infinite dimensional complex Hilbert space. We denote by $\lh$ and $\trh$ the Banach spaces of bounded operators and trace class operators on $\hi$, respectively. The set of quantum states (i.e.~positive trace one operators) is denoted by $\sh$.
In this paper, for simplicity, we treat only finite sets of measurement outcomes, while 
many statements can be easily extended also to infinite outcome sets.  
We denote by $\Omega$ (or $\Omega', \Omega_a$, etc.) a finite set of measurement outcomes.\footnote{This set is equipped with the natural $\sigma$-algebra $\mathcal{F}=2^{\Omega}$ containing all subsets of $\Omega$. 
Thus $X\subseteq \Omega$ is equivalent to $X\in \mathcal{F}$ in this paper, while the latter should be employed in treating infinite outcome set $\Omega$.}

\section{Incompatible devices}\label{sec:incompatible-devices}

\subsection{Input-output devices}\label{sec:devices}

A quantum state is described by a density matrix, while a classical state is described by a probability distribution.
By a \emph{quantum device} (or shortly device), usually denoted by $\dev$, we mean an apparatus that takes a quantum input and produces either a classical output ($\rightarrow$~observable), a quantum output ($\rightarrow$~channel) or both ($\rightarrow$~instrument) --- see Fig.~\ref{fig:device}.
We can also consider probabilistic output, which means that an output is obtained with some probability that can be less than $1$.
Again, the output can be either classical ($\rightarrow$~effect) or quantum ($\rightarrow$~operation).

\begin{figure}
\begin{center}
\includegraphics[width=6.0cm]{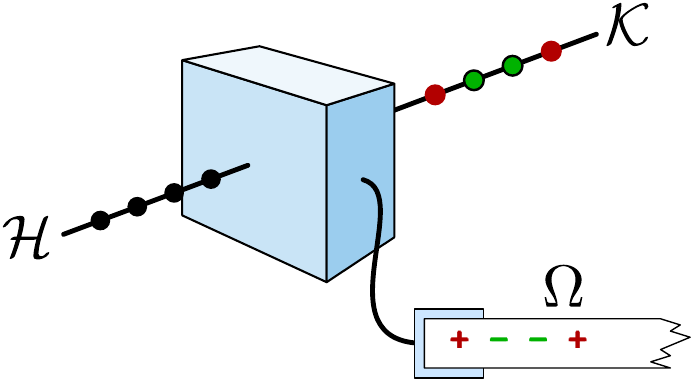}
\end{center}
\caption{\label{fig:device} A quantum device is an apparatus having Hilbert space $\hi$ as an input and as an output either a Hilbert space $\hik$, or a measurement outcome set $\Omega$, or both.
}
\end{figure}

We denote by $\hi,\hik$ two fixed Hilbert spaces associated with the input and output systems, respectively. 
The precise definitions of the previously mentioned five quantum devices are the following (see e.g. \cite{MLQT12}).

\begin{itemize}
\item An \emph{effect} is an operator $E\in\lh$ satisfying $0\leq E \leq \id$.
\item An \emph{observable} is a map $\A:\Omega\to\lh$ such that $\A(x)$ is an effect for all $x\in\Omega$ and $\sum_{x\in\Omega} \A(x)=\id$.
We denote $\A(X)\equiv \sum_{x\in X} \A(x)$ for every $X\subseteq\Omega$.
\item An \emph{operation} is 
in the Heisenberg picture a normal completely positive map $\Phi\hh:\lk\to\lh$ satisfying $\Phi\hh(\id_\hik)\leq \id_\hi$.
An operation in the Schr\"odinger picture 
is a completely positive map 
$\Phi\ss: \trh \to \trk$ satisfying 
$\tr{\Phi\ss(\varrho)}\leq \tr{\varrho}$ for all $\varrho \in \sh$. 
\item A \emph{channel} is an operation $\Lambda$ that satisfies
 $\Lambda\hh(\id_\hik) = \id_\hi$ in the Heisenberg picture, or
 $\tr{\Lambda\ss(\varrho)}=\tr{\varrho}$ 
for all $\varrho \in \sh$  in the Schr\"odinger picture.
\item An \emph{instrument} is in the Heisenberg picture a map $\I\hh:\Omega\times\lk\to\lh$
such that each $\I\hh(x,\cdot)$ is an operation and $\sum_{x\in\Omega} \I\hh(x,\cdot)$ is a channel. 
We denote $\I(X,\cdot)\equiv \sum_{x\in X} \I(x,\cdot)$ for every 
$X\subseteq \Omega$. 
An instrument in the Schr\"odinger picture is 
a map $\I\ss : \Omega\times \trh\to \trk$ 
such that each $\I\ss(x,\cdot)$ is an operation and 
$\I\ss({\Omega},\cdot)$ is a channel. 

\end{itemize}

We will use superscripts $\ss$ and $\hh$ for the Schr\"odinger 
and Heisenberg pictures, respectively. 
If a statement or equation is identical in both Heisenberg and Schr\"odinger pictures, then we may leave the superscript out. Also we will often use $\cdot$ as a placeholder for appropriate variable when it is evident. 

It is clear that an instrument is the most comprehensive description among the five devices since it has both quantum and classical output.
We can thus introduce the notion ``part of an instrument'' for any of the above five devices. 

\begin{definition}\label{def:part-of-i-1}
Let $\I$ be an instrument. We say that:
\begin{itemize}
\item an effect $E$ is part of $\I$ if there exists a set $X\subseteq\Omega$ such that
\begin{align}\label{eq:effect}
E =  \I\hh(X,\id) \, ;
\end{align}
\item an operation/channel $\Phi$ is part of $\I$ if there exists a set $X\subseteq\Omega$ such that
\begin{equation}
\Phi(\cdot) = \I(X,\cdot) \, .
\end{equation}
\end{itemize}
\end{definition}

If an effect/operation is part of an instrument, then we can think that the former gives a \emph{partial} mathematical description of some quantum apparatus, while the instrument gives a \emph{complete} mathematical description of the apparatus in question. 

The following simple fact will be used on several occasions in our investigation.

\begin{lemma}\label{lemma:span}
Any operator $T\in\lh$ can be written as a linear combination of four effects.
\end{lemma}

\begin{proof}
We can decompose $T$ into two self-adjoint operators $T=T_R + i T_I$, where $T_R=\half (T+T^*)$ and $T_I=\frac{1}{2i} (T-T^*)$.
Further, any self-adjoint operator $S$ can be written as a difference of two positive operators $S=S_+-S_-$,  where $S_+=\half (\no{S}\id + S)$ and $S_-=\half (\no{S}\id - S)$. 
Finally, any positive operator $P$ can be written as a scalar multiple of an effect since $P=\no{P} \left( P/\no{P} \right)$ and $0\leq P/\no{P}\leq\id$.
\end{proof}

\begin{proposition}\label{prop:channel-part-instrument}
A channel $\Lambda$ is part of an instrument $\I$ with an outcome set $\Omega$
if and only if $\Lambda(\cdot)=\I(\Omega,\cdot)$ holds. 
\end{proposition}

\begin{proof}
The ``if'' part is trivial. 
Let us consider the ``only if'' part. 
Suppose $\Lambda\hh(\cdot)=\I\hh(X,\cdot)$ for some $X\subset \Omega$. 
We make a counter assumption that $\I\hh(X,\cdot) \neq \I\hh(\Omega,\cdot)$. 
This implies that there exists $T\in\lk$ such that $\I\hh(\Omega\setminus X,T)\neq 0$. 
By Lemma~\ref{lemma:span} this $T$ can be taken so as to satisfy $0\leq T\leq \id$.
It follows that $\I\hh(\Omega\setminus X,\id)\neq 0$. 
Hence
\begin{align*}
\id & = \I\hh(X,\id) + \I\hh(\Omega\setminus X,\id) = \Lambda\hh(\id)+\I\hh(\Omega\setminus X,\id) \\ 
& = \id + \I\hh(\Omega\setminus X,\id) \neq \id \, .
\end{align*}
This contradiction means that the counter assumption is false.
\end{proof}

Observables and instruments do not describe single events but collections of possible events. 
While for effects, operations and channels the labeling of measurement outcomes is irrelevant, for observables and instruments this is part of their description.
Since the measurement outcomes can be regrouped and relabeled after the measurement is performed, we include a \emph{pointer function} into our description.

\begin{definition}\label{def:part-of-i-2}
Let $\I$ be an instrument. We say that:
\begin{itemize}
\item an observable $\A$ with an outcome set $\Omega'$ is part of $\I$ if there exists a function $f:\Omega\to\Omega'$ such that
\begin{align}\label{eq:part-observable}
\A(X) =  \I\hh(f^{-1}(X),\id) 
\end{align}
for all $X\subseteq\Omega'$;
\item an instrument $\I'$  with an outcome set $\Omega'$ is part of $\I$ if there exists a function $f:\Omega\to\Omega'$ such that
\begin{equation}\label{eq:part-instrument}
\I'(X,\cdot) = \I(f^{-1}(X),\cdot)
\end{equation}
for all $X\subseteq\Omega'$.
\end{itemize}
\end{definition}

Let us remark that since $\Omega$ and $\Omega'$ are finite sets, it is enough to require \eqref{eq:part-observable} and \eqref{eq:part-instrument} for all singleton sets $X=\{x\}$, $x\in\Omega'$.
The equality for other sets then follows from the fact that $f^{-1}(X\cup Y)=f^{-1}(X)\cup f^{-1}(Y)$ for all subsets $X,Y\subseteq\Omega'$.

\begin{example}
(\emph{Every device is part of some instrument.})\label{ex:partofinstrument}
For any given device we can construct an instrument that has that device as its part.
The following simple constructions also show that there are always uncountably many different instruments with that property.

Let $E$ be an effect.
We fix a state $\varrho_0$ and define an instrument $\I$ with the outcome set $\{0,1\}$ by
$\I\hh(0,T)=\tr{\varrho_0 T}E$ and $\I\hh(1,T)=\tr{\varrho_0 T}(\id -E)$. 
Then $\I\hh(0,\id)= E$ and thus $E$ is part of $\I$.

Let $\A$ be an observable with an outcome set $\Omega$.
We fix a state $\varrho_0$ and define an instrument $\I$ with the outcome set $\Omega$ by $\I\hh(x,T)=\tr{\varrho_0 T}\A(x)$. 
Then $\I\hh(x,\id)= \A(x)$ and thus $\A$ is part of $\I$.

Let $\Phi$ be an operation.
We fix a state $\varrho_0$ and define an instrument $\I$ with an outcome set $\{0,1\}$ by $\I\hh(0,T)=\Phi\hh(T)$ 
and 
$\I\hh(1,T)=\tr{\varrho_0 T} (\id -\Phi\hh(\id))$.

In all of the previous three instances we are free to choose $\varrho_0$, hence we have uncountably many different instruments that have given device as its part (and these still need not be all the possibilities).

Let $\Lambda$ be a channel. 
Then the above construction for operations becomes trivial and gives only a single instrument. 
An alternative construction gives again infinitely many different instruments.
Namely, we fix a probability distribution $p$ on an outcome set $\Omega$.
We define an instrument $\I$ with the outcome set $\Omega$ by $\I(x,\cdot) = p(x)\Lambda(\cdot)$.
\end{example}

\subsection{Incompatiblility}\label{sec:usual}

The key idea behind the notion of compatibility is the fact that we can duplicate the classical measurement outcome data and process it in various different ways. 
In this way, two quite different devices can be parts of a single instrument.
The interesting cases are those where this kind of duplication cannot help in implementing two different devices.
Then the devices are incompatible and they manifest a significant feature of quantum theory. 

\begin{definition}\label{def:comp}
Two devices $\dev_1$ and $\dev_2$ are \emph{compatible} if 
there exists an instrument that has both $\dev_1$ and $\dev_2$ as its parts; otherwise $\dev_1$ and $\dev_2$ are \emph{incompatible}.
\end{definition}

This definition is a direct generalization of the notion of coexistence of two operations \cite{HeReStZi09}.
We also have the following results, which conclude that compatibility generalizes the notions of coexistence of effects \cite{SEO83} and joint measurability of observables \cite{LaPu97}.

\begin{proposition}
The compatibility of observables and effects reduces to the standard relations:
\begin{itemize}
\item[(a)] Two effects $E_1$ and $E_2$ are compatible if and only if they are coexistent, i.e.~there exists an observable $\G$ with an outcome set $\Omega$ such that 
\begin{equation}\label{eq:coex}
E_1=\G(X_1) \, , \qquad E_2=\G(X_2)
\end{equation}
for some $X_1,X_2\subseteq\Omega$.
\item[(b)] Two observables $\A_1$ and $\A_2$, with outcome sets $\Omega_1$ and $\Omega_2$, are compatible if and only if they are jointly measurable, i.e.~there exists an observable $\G$ on $\Omega_1\times\Omega_2$ such that
\begin{equation}\label{eq:jointobs}
\G(X\times\Omega_2) = \A_1(X) \, , \qquad \G(\Omega_1\times Y) = \A_2(Y) 
\end{equation}
for all $X\subseteq\Omega_1, Y\subseteq\Omega_2$.
\end{itemize}
\end{proposition}

\begin{proof}
\begin{itemize}

\item[(a)] Suppose that $E_1$ and $E_2$ are coexistent, hence there exists an observable $\G$ with an outcome set $\Omega$ that satisfies \eqref{eq:coex}.
We fix a state $\varrho_0$ and define an instrument $\I$ by
\begin{equation}\label{eq:coex-ins}
\I\hh(X,\cdot)=\tr{\varrho_0\, \cdot\,} \G(X) \, , \quad X\subseteq\Omega \, .
\end{equation}
Then $\I\hh(X,\id)=\G(X)$ and therefore both $E_1$ and $E_2$ are parts of $\I$, hence compatible.

Suppose then that $E_1$ and $E_2$ are compatible, hence there exists an instrument $\I$ such that $\I\hh(X_1,\id)=E_1$ and $\I\hh(X_2,\id)=E_2$ for some $X_1,X_2\subseteq\Omega$.
We set $\G(X):=\I\hh(X,\id)$, and then $\G(X_1)=E_1$ and $\G(X_2)=E_2$.
This means that $E_1$ and $E_2$ are coexistent.

\item[(b)] This proof is very similar to the proof of (a). 
Suppose that $\A_1$ and $\A_2$ are jointly measurable.
By definition, there exists an observable $\G$ that satisfies \eqref{eq:jointobs}.
We use this $\G$ to define an instrument $\I$ by \eqref{eq:coex-ins}.
Then both $\A_1$ and $\A_2$ are parts of $\I$, hence compatible.

Suppose then that $\A_1$ and $\A_2$ are compatible.
By definition, there exists an instrument $\I$ such that both $\A_1$ and $\A_2$ are parts of $\I$.
The observable $\G(\cdot)=\I\hh(\cdot,\id)$ 
gives both $\A_1$ and $\A_2$ as its functions. 
By Theorem 3.1 in \cite{LaPu97} this is equivalent to joint measurability of $\A_1$ and $\A_2$.
\end{itemize}
\end{proof}

In addition to generalizing the usual concepts of joint measurability of observables and coexistence of effects, Definition \ref{def:comp} gives a way to speak about compatibility between two devices of different types.
First we make some simple observations.

\begin{proposition}\label{prop:simple}
If an operation $\Phi$ is compatible with a device $\dev$, then the effect $\Phi\hh(\id)$ is compatible with $\dev$.
\end{proposition}

\begin{proof}
If an operation $\Phi$ is part of an instrument $\I$, then $\Phi^{\hh}(\cdot)=\I^{\hh}(X,\cdot)$ for some set $X\subseteq\Omega$.
Then also the effect $\I^{\hh}(X,\id)=\Phi^{\hh}(\id)$ is part of $\I$.
\end{proof}

The reverse of the implication in Prop.~\ref{prop:simple} is not valid as the constraints on the compatibility of two operations are typically more strict as the constraints on the related effects. 
This difference has been demonstrated in \cite{HeReStZi09} where it was shown that two operations $\varrho\mapsto A^{1/2}\varrho A^{1/2}$ and $\varrho\mapsto B^{1/2}\varrho B^{1/2}$, where $A$ and $B$ are effects, are compatible either if $A$ is a multiple of $B$ or if $A+B\leq\id$. 
However, the compatibility of two effects is obviously not restricted just to these relations.
For instance, two commuting effects are always compatible \cite{SEO83}.

The physical explanation of the fact that operations are ``not so easily'' compatible as the related effects is that operations give a more detailed description than effects. 
Since we are asking whether two mathematical descriptions can correspond to the same device, it is more likely that two coarser descriptions have this property than two finer descriptions.

The correct reverse of Prop.~\ref{prop:simple} is the following.

\begin{proposition}\label{prop:simple-2}
If an effect $E$ is compatible with a device $\dev$, then there exists an operation $\Phi$ such that $E=\Phi\hh(\id)$ and $\Phi$ is compatible with $\dev$.
\end{proposition}

\begin{proof}
If an effect $E$ is part of an instrument $\I$, then $E=\I^{\hh}(X,\id)$ for some set $X\subseteq\Omega$.
Then also the operation $\Phi$, defined as $\Phi\hh(\cdot):=\I^{\hh}(X,\cdot)$ is part of $\I$.
\end{proof}

Intuitively, one can always join two ``disjoint'' descriptions into a total description since there cannot be a conflict between them.
In mathematical terms, this leads to the following statements.

\begin{proposition}\label{prop:sufficient}
The following conditions are sufficient for compatibility:
\begin{itemize}
\item[(a)] Two effects $E_1$ and $E_2$ are compatible if $E_1+E_2 \leq \id$.
\item[(b)] An operation $\Phi$ and an effect $E$ are compatible if $\Phi\hh(\id) + E\leq \id$.
\item[(c)] Two operations $\Phi_1$ and $\Phi_2$ are compatible if $\Phi_1\hh(\id)+\Phi_2\hh(\id) \leq \id$.
\end{itemize}
\end{proposition}

\begin{proof}
The points (a) and (b) follow from (c) and Proposition \ref{prop:simple}.
To prove (c), suppose that $\Phi_1\hh(\id)+\Phi_2\hh(\id) \leq \id$.
We fix a state $\xi$ and define a ternary instrument $\I$ by 
\begin{align*}
\I\hh(1,\cdot) & = \Phi_1\hh(\cdot) \, , \quad
\I\hh(2,\cdot) = \Phi_2\hh(\cdot) \, , \\
\I\hh(3,\cdot) & =  \tr{\xi \cdot} (\id- \Phi_1\hh(\id)-\Phi_2\hh(\id)) \, .
\end{align*}
Both $\Phi_1$ and $\Phi_2$ are parts of this instrument, hence they are compatible.
\end{proof}

It is a well known fact that two effects $E,F$ are compatible (i.e.~coexistent) if they commute, and that a projection $P$ is compatible with an effect $E$ if and only if they commute; see e.g.~p.~120 in \cite{SEO83} or \cite{BuHe08}.
Proposition~\ref{prop:channelsharpeffect} is an analogous result for a projection and an operation.
(The ``only if'' part of (b) can also be inferred from Theorem 3 in \cite{Ozawa01pra}.)

\begin{proposition}\label{prop:channelsharpeffect}
Let $\Phi\hh:\lk\to\lh$ be an operation.
\begin{itemize}
\item[(a)] An effect $E\in\lh$ is compatible with $\Phi$ if 
\begin{eqnarray*}
[\Phi\hh(T),E]=0   \qquad \forall T\in \lk \, .
\end{eqnarray*} 
\item[(b)] A projection $P\in\lh$ is compatible with $\Phi$ if and only if
\begin{eqnarray*}
[\Phi\hh(T),P]=0   \qquad \forall T\in \lk \, .
\end{eqnarray*}
\end{itemize} 
\end{proposition}

\begin{proof}
\begin{itemize}
\item[(a)]
Suppose $[\Phi\hh(T),E]=0$ for all $T\in\lk$.
This implies that $[\Phi\hh(T),\sqrt{E}]=0$ and $[\Phi\hh(T),\sqrt{\id-E}]=0$ for all $T\in\lk$.
Then $\Phi\hh$ can be written as 
\begin{eqnarray*}
\Phi\hh(T) &=& E \Phi\hh(T)+ (\id-E)\Phi\hh(T) \\
&=& \sqrt{E} \sqrt{E} \Phi\hh(T)+ \sqrt{\id-E}\sqrt{\id-E}\Phi\hh(T) \\
&=& \sqrt{E}\Phi\hh(T)\sqrt{E}+\sqrt{\id-E}\Phi\hh(T) \sqrt{\id-E} \, .
\end{eqnarray*} 
We fix a state $\eta$ and set 
\begin{align*}
& \I\hh(0,T):=\sqrt{E}\Phi\hh(T)\sqrt{E} \\ 
& \I\hh(1,T):=\sqrt{\id-E}\Phi\hh(T) \sqrt{\id-E} \\
& \I\hh(2,T):=\tr{\eta T} (\id - \Phi\hh(\id)) \, .
\end{align*}
Then $\I$ is a ternary instrument with $\I\hh(\{0,1\},\cdot)=\Phi\hh(\cdot)$ and 
$\I\hh(0,\id)=E$.
Hence, $\Phi$ and $E$ are compatible.

\item[(b)] We need to prove the ``only if'' part of the statement. 
Suppose that $P$ and $\Phi$ are compatible. Thus there exists an instrument $\I$ satisfying $\I\hh(X,\id)=P$
and $\I(Y,\cdot)=\Phi(\cdot)$ for some $X,Y\subseteq\Omega$. 
We make a counter assumption that there exists $T\in\lk$ such that
\begin{equation}
\label{eq:noncommute}
[\Phi\hh(T), P]\neq 0 \, .
\end{equation}
By Lemma \ref{lemma:span} this $T$ can be assumed to be an effect. 
We write $Y$ as a disjoint union $Y=Y_1 \cup Y_2$, where $Y_1=Y\cap X$ and $Y_2=Y\cap (\Omega\setminus X)$.
From \eqref{eq:noncommute} follows that
\begin{align}
\textrm{\ } \qquad & [\I\hh(Y_1,T), P]\neq 0
\label{eq:00} \\
\textrm{or} \qquad & [\I\hh(Y_2,T),P]\neq 0 \, .
\label{eq:01}
\end{align}
Since $Y_1\subseteq X$ and $0\leq T\leq \id$, we have 
\begin{equation*}
0\leq \I\hh(Y_1,T)\leq\I\hh(X,T)\leq \I\hh(X,\id)=P \, .
\end{equation*}
A positive operator below a projection commutes with that projection, thus \eqref{eq:00} cannot hold.
In a similar way we obtain
\begin{equation*}
0\leq \I\hh(Y_2,T)\leq\I\hh(\Omega\setminus X,T)\leq \I\hh(\Omega\setminus X,\id)=\id-P \, .
\end{equation*}
Thus $\I\hh(Y_2,T)$ commutes with $\id-P$, hence also with $P$ and \eqref{eq:01} cannot hold.
Therefore, $\Phi$ and $P$ must be incompatible if \eqref{eq:noncommute} holds.
\end{itemize}
\end{proof}

The preceding result leads to the following observation. 

\begin{proposition}\label{prop:channel-trivial}
If a channel $\Lambda$ is compatible with all projections, then $\Lambda$ is a contraction channel,
i.e.~$\Lambda\hh(T)=\tr{\eta T}\id$ for some fixed state $\eta$. 
\end{proposition}

\begin{proof}
Suppose $\Lambda$ is compatible with all projections and let $T\in\lh$.
By Prop.~\ref{prop:channelsharpeffect} we have $[\Lambda\hh(T),P]=0$ for all projections $P$, hence $\Lambda\hh(T)=c(T) \id$ for some number $c(T)\geq 0$.
Because $\Lambda\hh(\cdot)$ is a unital normal positive linear map, $c(\cdot)$ is a 
unital normal positive linear functional. 
Hence, $c(\cdot)$ is identified by some state $\eta$ on $\mathcal{L(H)}$ via the trace formula, $c(T)=\tr{\eta T}$.
\end{proof}

A device $\dev$ is called \emph{trivial} if it is compatible with all devices of the same kind.
A paradigmatic example is a coin tossing observable --- this gives a random outcome irrespective of the input state.
We can obviously toss a coin simultaneously with another measurement, and hence they must be compatible.
In the following we demonstrate the compatibility relation by recalling more about trivial devices.

\begin{example}(\emph{Trivial devices})\label{ex:trivial-devices}
An effect $E$ is trivial if and only if $E=e\id$ for some number $0\leq e\leq 1$.
This is a simple consequence of the fact that an effect and a projection are coexistent if and only if they commute.
It also follows that an observable $\A$ is trivial if and only if $\A(x)$ is a trivial effect for each $x$.
The only trivial operation is the null operation $\varrho\mapsto 0$ as shown in \cite{HeReStZi09}.
It follows that there are no trivial channels nor trivial instruments.

The general definition of compatibility allows us to consider also devices that are compatible with all devices of some different type.

Any contraction channel $\varrho\mapsto\eta$, where $\eta$ is a fixed state, is compatible with every observable.
Namely, if $\A$ is an observable, we define an instrument by 
\begin{equation}
\I\hh(x,T)= \tr{\eta T} \A(x) \, .
\end{equation}
The contraction channels are the only channels with the property that they are compatible with every observable. 
This is the result of Prop.~\ref{prop:channel-trivial}.

Any trivial observable is compatible with every channel.
Namely, if $x\mapsto p(x)\id$ is a trivial observable (here $p$ is a fixed probability distribution) and $\Lambda$ is a channel, we define an instrument by formula
\begin{equation}
\I\hh(x,\cdot)= p(x) \Lambda\hh(\cdot) \, .
\end{equation}
The trivial observables are the only observables that are compatible with every channel.
This follows from the fact that any non-trivial observable disturbs some state \cite{Busch09}, which means that the identity channel that preserves all input states cannot be compatible with this observable.

The preceding two statements are information-disturbance counterparts in the compatibility language.
While the latter one states that there is no information without disturbance, the former one shows that complete information means complete destruction.
\end{example}

\subsection{Strong incompatibility}\label{sec:strong}

If two devices $\dev_1$ and $\dev_2$ are incompatible, there is no single instrument that would give both $\dev_1$ and $\dev_2$ as its parts. 
Of course, we can always separately implement $\dev_1$ and $\dev_2$ with two different instruments.
We can then ask whether these two instruments need to be completely different or whether they have some similarity. 
This motivates the following definitions.

\begin{definition}
Two devices $\dev_1$ and $\dev_2$ are \emph{weakly compatible} if there exist two instruments $\I_1$ and $\I_2$ such that $\dev_1$ is part of $\I_1$ and $\dev_2$ is part of $\I_2$, and that $\I_1(\Omega,\cdot)=\I_2(\Omega,\cdot)$.
Otherwise we say that $\dev_1$ and $\dev_2$ are \emph{strongly incompatible}.
\end{definition}

Clearly, compatible devices $\dev_1$ and $\dev_2$ are weakly compatible.
Or in other words, strongly incompatible devices are incompatible.
In some case strong incompatibility can be either equivalent to incompatibility or impossible.
For this we have the following simple observations.

\begin{proposition}\label{prop:elementary}
\begin{itemize}
\item[(a)] A channel $\Lambda$ is strongly incompatible with another device $\dev$ if and only if they are incompatible.
\item[(b)] All pairs of observables/effects are weakly compatible.
\end{itemize}
\end{proposition}

\begin{proof}
\begin{itemize}
\item[(a)] Follows from Prop.~\ref{prop:channel-part-instrument}.
\item[(b)] Let $\A_1$ and $\A_2$ be two observables.
We fix a state $\eta$ and define instruments $\I_1$ and $\I_2$ by
\begin{equation*}
\I_j\hh(x,\cdot)=\tr{\eta \cdot} \A_j(x) \, .
\end{equation*}
Then $\I_1\hh(\Omega,\cdot)=\I_2\hh(\Omega,\cdot)=\tr{\eta \cdot} \id $.
Hence, $\A_1$ and $\A_2$ are weakly compatible.
The weak compatibility in the other two cases (effect-effect and effect-observable) can be proved in a similar way.
\end{itemize}
\end{proof}

We also have analogous statements as in Prop.~\ref{prop:simple} and Prop.~\ref{prop:simple-2}.

\begin{proposition}\label{prop:simple-3}
\begin{itemize}
\item[(a)] If an operation $\Phi$ is weakly compatible with device $\dev$, then the effect $\Phi\hh(\id)$ is weakly compatible with $\dev$.
\item[(b)] If an effect $E$ is weakly compatible with a device $\dev$, then there exists an operation $\Phi$ such that $E=\Phi\hh(\id)$ and $\Phi$ is weakly compatible with $\dev$.
\end{itemize}
\end{proposition}

\begin{proof}
The proofs are similar to the proofs of Prop.~\ref{prop:simple} and Prop.~\ref{prop:simple-2}.
\end{proof}

For two operations $\Phi_1$ and $\Phi_2$, we denote $\Phi_1\leq\Phi_2$ if there exists an operation $\Phi'$ such that $\Phi_1 + \Phi'= \Phi_2$.
This relation is a partial order in the set of operations, and has been studied e.g.~in \cite{Arveson69}.
Clearly, to see whether $\Phi_1\leq\Phi_2$ holds, we only need to check if the mapping $\Phi_2-\Phi_1$ is completely positive.

Let us remark that even if $\Phi_1$ and $\Phi_2$ are completely positive, their difference $\Phi_2-\Phi_1$ can be positive without being completely positive.
For instance, the linear maps $\Phi_1,\Phi_2:\mathcal{T}(\complex^2)\to\mathcal{T}(\complex^2)$, defined by
\begin{equation}
\Phi\ss_1(\varrho) = \tr{\varrho}\id/3 \, , \quad  \Phi\ss_2(\varrho)= (\tr{\varrho}\id + \varrho^T )/3  \, , 
\end{equation}
are operations (here $\cdot^T$ is the transposition in some fixed bases).
The difference of $\Phi_2-\Phi_1$ is a multiple of the transposition map, thus positive but not completely positive.

If $\Phi_1$ and $\Phi_2$ are comparable (i.e.~$\Phi_1\leq\Phi_2$ or $\Phi_2\leq\Phi_1$), then they are compatible. 
We can see this by defining an instrument $\I$ with the outcome set $\Omega=\{1,2,3\}$ as follows (assuming that $\Phi_1\leq\Phi_2$):
\begin{align*}
& \I\hh(1,\cdot)=\Phi\hh_1(\cdot) \, ,\quad \I\hh(2,\cdot)=(\Phi\hh_2-\Phi\hh_1)(\cdot) \, , \\
& \I\hh(3,\cdot)=\tr{\cdot\xi} (\id-\Phi_2\hh(\id)) \, , 
\end{align*}
where $\xi$ is some fixed state.

An operation $\Phi$ is called \emph{pure} if it can be written in the form $\Phi(\cdot) = W\cdot W^\ast$ for some bounded operator $W$.
We recall from Prop.~4 in \cite{HeReStZi09} that two pure operations $\Phi_1$ and $\Phi_2$ are compatible if and only if they are comparable or their sum $\Phi_1+\Phi_2$ is an operation.
Unlike for pure operations, generally compatibility of two operations cannot be expressed as a simple condition using the above partial order. 
These conditions are only sufficient, not necessary.

Weak compatibility has a clear characterization in terms of this partial order.
Namely, the weak compatibility of two operations reduces to the requirement that the operations have a common upper bound.

\begin{proposition}\label{prop:channel_for_wc}
Let $\Phi_1$ and $\Phi_2$ be two operations.
The following conditions are equivalent:
\begin{itemize}
\item[(i)]  $\Phi_1$ and $\Phi_2$ are weakly compatible.
\item[(ii)] There exists a channel $\Lambda$ such that $\Phi_1\leq\Lambda$ and $\Phi_2\leq\Lambda$.
\item[(iii)] There exists an operation $\Phi$ such that $\Phi_1\leq\Phi$ and $\Phi_2\leq\Phi$.
\end{itemize}
\end{proposition}

\begin{proof}
(i)$\Rightarrow$(ii): Assuming (i), there exist two instruments $\I_1$ 
on $\Omega_1$ and $\I_2$ on $\Omega_2$ such that 
$\I_1(X_1,\cdot)=\Phi_1(\cdot)$ for some $X_1\subseteq \Omega_1$, $\I_2(X_2,\cdot)=\Phi_2(\cdot)$ for some  $X_2 \subseteq \Omega_2$, and $\I_1(\Omega_1,\cdot) =\I_2(\Omega_2,\cdot)\equiv\Lambda(\cdot)$, where $\Lambda$ is a channel. 
It is now clear, that
\[
\Phi_a(\cdot)=\I_a(X_a,\cdot)\leq\I_a(X_a,\cdot)+\I_a(\Omega_a\setminus X_a,\cdot)=\I_a(\Omega_a,\cdot)=\Lambda(\cdot)
\]
for $a=1,2$. 

(ii)$\Rightarrow$(iii): Every channel is an operation, hence (ii) implies (iii) trivially.

(iii)$\Rightarrow$(i): We have $\Phi_1\leq\Phi$ and $\Phi_2\leq\Phi$, which means that $\bar{\Phi}_1:=\Phi-\Phi_1$  and $\bar{\Phi}_2:=\Phi-\Phi_2$ are operations. 
We define an instrument $\I_1$ with the outcome set $\Omega=\{1,2,3\}$ as follows:
\begin{equation*}
\I\hh_1(1,\cdot)=\Phi\hh_1(\cdot) \, ,\quad \I\hh_1(2,\cdot)=\bar{\Phi}\hh_1(\cdot) \, , \quad \I\hh_1(3,\cdot)=\tr{\,\cdot\,\xi} (\id-\Phi\hh(\id)) \, , 
\end{equation*}
where $\xi$ is some fixed state.
In a similar way we define an instrument $\I_2$ related to the operation $\Phi_2$,
\begin{equation*}
\I\hh_2(1,\cdot)=\Phi\hh_2(\cdot) \, ,\quad \I\hh_2(2,\cdot)=\bar{\Phi}\hh_2(\cdot) \, , \quad \I\hh_2(3,\cdot)=\tr{\,\cdot\,\xi} (\id-\Phi\hh(\id)) \, . 
\end{equation*}
Since $\I_1(\Omega,\cdot)=\I_2(\Omega,\cdot)$ we conclude that $\Phi_1$ and $\Phi_2$ are weakly compatible.
\end{proof}

The following statement follows immediately. 

\begin{proposition}\label{prop:order-channel-operation}
A channel $\Lambda$ is weakly compatible (hence also compatible) with an operation $\Phi$ if and only if $\Phi\leq \Lambda$ holds.  
\end{proposition}

For some operations $\Phi$ of specific type, it is easy to write down explicitly all channels $\Lambda$ satisfying $\Phi\leq\Lambda$.
In the next proof and also later we will use the fact that if $\Phi$ is an operation such that $\Phi\hh(\id)$ is a rank-1 operator, then $\Phi$ is  of the form $\Phi\ss(\cdot)=\tr{\,\cdot\,\Phi\hh(\id)} \xi$ for some state $\xi$; see Prop.~8 in \cite{HeWo10}.

\begin{proposition}\label{prop:order-rank-1}
Let $\Phi$ be an operation such that $\id-\Phi\hh(\id)$ is a rank-1 operator.
Then a channel $\Lambda$ satisfies $\Phi\leq\Lambda$ if and only if 
\begin{equation}\label{eq:order-rank-1}
\Lambda\ss(\cdot) = \Phi\ss(\cdot) + (1-\tr{\Phi\ss(\cdot)}) \xi
\end{equation}
for some state $\xi$.
\end{proposition}

\begin{proof}
Clearly, if \eqref{eq:order-rank-1} holds, then $\Phi\leq\Lambda$.  

Suppose then that a channel $\Lambda$ satisfies $\Phi\leq\Lambda$, i.e.~there exists an operation $\Phi'$ such that $\Phi+\Phi'=\Lambda$.
In particular, $\Phi\hh(\id) + {\Phi'}\hh(\id)=\Lambda\hh(\id)=\id$.
We denote $E\equiv\id-\Phi\hh(\id)={\Phi'}\hh(\id)$.
Since $E$ is a rank-1 operator, the operation $\Phi'$ is of the form
\begin{equation}
{\Phi'}\ss(\varrho) =\tr{\varrho E} \xi
\end{equation}
for some state $\xi$.
Inserting this into $\Lambda=\Phi+\Phi'$ we conclude \eqref{eq:order-rank-1}.
\end{proof}

\section{Mathematical formulations of incompatibility}\label{sec:formulations}

In this section we formulate the incompatibility relations using first the Stinespring representation and then Kraus operators. 

\subsection{Incompatibility in terms of Stinespring dilation}\label{sec:appendix}

Let us begin with the well-known standard Stinespring representation theorem. 
(See for e.g.~\cite{CBMOA03}.)

\begin{theorem}(Stinespring representation)\label{thm:stine} 
Let $\Phi\hh: \mathcal{L(K)} \to \mathcal{L(H)}$ be 
an operation. There exist a Hilbert space $\hik'$ and 
an operator $V:\mathcal{H} \to \mathcal{K}\otimes \hik'$ 
satisfying 
\begin{eqnarray*}
\Phi\hh(T)=V^* (T\otimes \id)V
\end{eqnarray*}
for all $T\in \mathcal{L(K)}$. The doublet $(\hik', V)$ is called 
the Stinespring representation of $\Phi$. 
It holds that $\Vert \Phi\Vert =\Vert \Phi\hh(\id)\Vert =
\Vert V\Vert^2$. 
In addition, 
if 
$(\mathcal{L(K)} \otimes \id) V \mathcal{H}$ is dense in 
$\mathcal{K}\otimes \hik'$, the representation is called 
minimal. The minimal representation exists and is determined uniquely 
up to unitary operations on $\hik'$. That is, 
if $(\hik'',V')$ is another minimal Stinespring representation, there exists 
a unitary operator $U:\hik' \to \hik''$ satisfying 
$V'=(\id \otimes U)V$. 
\end{theorem}

The ordering between operations 
can be expressed in terms of the Stinespring representation. 
The
following lemma is known as the Radon-Nikodym theorem for completely
positive maps \cite{Raginsky03}.
  
\begin{theorem}(Radon-Nikodym theorem for operations)\label{RN}
Let $\Phi_1\hh: \mathcal{L(K)} \to \mathcal{L(H)}$ be an operation.
We denote its minimal Stinespring representation by 
\begin{equation*}
\Phi_1\hh(T)=V^* (T\otimes \id)V,
\end{equation*}
where $V:\hi \to \hik\otimes \hik'$ is a linear operator. 
Let $\Phi_2\hh:\mathcal{L(K)}\to \mathcal{L(H)}$ be another operation. 
Then $\Phi_2\leq \Phi_1$ holds  
if and only if there exists an effect $E\in \mathcal{L(\hik')}$ such that 
\begin{equation*}
\Phi_2\hh(T)=V^* (T\otimes E)V
\end{equation*}
holds for every $T\in\lk$.
If $E$ exists, then it is unique. 
\end{theorem}

This statement leads to the following observations. 

\begin{proposition}
\label{prop:opopwc}(operation-operation weak compatibility)
Let $\Phi\hh_1$ and $\Phi\hh_2$ be two operations. 
They are weakly compatible if and only if 
there exist a Hilbert space $\hik'$, an isometry $V:\mathcal{H}
\to \hik \otimes \hik'$, and (not necessarily compatible)
 effects $E,F \in \mathcal{L(K')}$ 
satisfying 
\begin{eqnarray*}
\Phi\hh_1(T) &=& V^*(T\otimes E)V ,\\
\Phi\hh_2(T) &=& V^*(T\otimes F)V.
\end{eqnarray*}
\end{proposition}
\begin{proof}
``If'' part: 
It is easy to see that  $\Lambda\hh:\mathcal{L(K)}
\to \mathcal{L(H)}$ defined by 
$\Lambda\hh(T)=V^*(T\otimes \id)V$ is a channel 
satisfying $\Phi_1 \leq \Lambda$ and $\Phi_2\leq \Lambda$ and invoking linearity of $V$ concludes this part of the proof. 
\\
``Only if'' part: 
By Prop.~\ref{prop:channel_for_wc}, there exists a channel $\Lambda$ satisfying $\Phi_1\leq \Lambda$ and $\Phi_2 \leq \Lambda$. 
Then Theorem \ref{RN} is applied. 
\end{proof}

\begin{proposition}(channel-operation compatibility)
Let $\Lambda\hh: \mathcal{L(K)}\to \mathcal{L(H)}$ be 
a channel with a minimal Stinespring representation
\begin{eqnarray*}
\Lambda\hh(T)=V^*(T\otimes \id)V, 
\end{eqnarray*}
where $V:\mathcal{H}\to \mathcal{K}\otimes \hik'$
is an isometry. 
An operation $\Phi\hh: \mathcal{L(K)}\to \mathcal{L(H)}$ is compatible with $\Lambda\hh$ if and only if 
there exists an effect $E \in \mathcal{L(K')}$ such that 
\begin{equation*}
\Phi\hh(T)=V^* (T\otimes E)V
\end{equation*}
holds for every $T\in\lk$.
If $E$ exists, then it is unique. 
\end{proposition}

\begin{proof}
This follows from Prop.~\ref{prop:order-channel-operation} and Theorem \ref{RN}. 
\end{proof}

To discuss the compatibility between other combinations of devices, we have to generalize the Radon-Nikodym theorem. 
Since we are assuming that $\Omega$ is finite, the following result easily follows from Theorem \ref{RN}. 

\begin{proposition}\label{e(X)}(instrument-channel compatibility)
Let 
$\Lambda\hh: \mathcal{L(K)} \to \mathcal{L(H)}$ be a channel with a 
minimal Stinespring representation 
\begin{equation*}
\Lambda\hh(T)=V^* (T\otimes \id)V,
\end{equation*}
where $V:\mathcal{H} \to \mathcal{K}\otimes \hik'$ is an isometry. 
An instrument $\I\hh$ defined on $\Omega$ is 
compatible with $\Lambda\hh$ if and only if 
there exists an observable $\A$ on $\mathcal{L(\hik')}$ defined on 
$\Omega$ such that 
\begin{equation*}
\I\hh(x,T)=V^*(T\otimes \A(x))V
\end{equation*}
for every $x\in\Omega$ and $T\in\lk$. 
If $\A$ exists, then it is unique. 
\end{proposition}

The result leads to the following characterization of compatible operations. 

\begin{proposition}(operation-operation compatibility)
Let $\Phi\hh_1$ and $\Phi\hh_2$ be operations 
$\mathcal{L(K)}\to \mathcal{L(H)}$. 
They are compatible if and only if 
there exist a Hilbert space $\hik'$, an isometry $V:\mathcal{H}
\to \hik \otimes \hik'$, and compatible effects $E,F \in \mathcal{L(K')}$ 
satisfying 
\begin{eqnarray*}
\Phi\hh_1(T)=V^*(T\otimes E)V\\
\Phi\hh_2(T)=V^*(T\otimes F)V
\end{eqnarray*}
for all $T\in\lk$.
\end{proposition}
\begin{proof}
Let us begin with ``if part''. 
Thanks to the compatibility between $E$ and $F$, 
there is an observable $\A$ on $\Omega$ and 
$X_1,X_2\subseteq \Omega$ such that 
$E=\A(X_1)$ and $F=\A(X_2)$ hold. 
Following the proof of ``if'' part of Proposition \ref{e(X)}, 
one can show that 
$\I(\cdot,\cdot)$ defined by 
$\I(X,T):=V^*(T\otimes \A(X))V$ 
for each $X\subseteq \Omega$ and $T\in \mathcal{L(K)}$ 
is an instrument although this representation 
is not necessarily the minimal representation. 
Because $\I(X_1,\cdot)=\Phi\hh_1(\cdot)$ and 
$\I(X_2,\cdot)=\Phi\hh_2(\cdot)$ hold, $\Phi\hh_1$ and 
$\Phi\hh_2$ are compatible. 
 \\
To prove ``only if'' part'' assume that $\Phi\hh_1$ and $\Phi\hh_2$ are 
compatible. Then there exists an instrument 
$\I(\cdot,\cdot)$ on $\Omega$ and $X_1,X_2\subset \Omega$ 
satisfying 
$\I(X_1,\cdot)=\Phi\hh_1(\cdot)$ and $\I(X_2,\cdot)
=\Phi\hh_2(\cdot)$. 
The instrument $\mathcal I$ is compatible with the channel
$\mathcal{I}(\Omega,\cdot)$. The ``only if'' part of Prop.~\ref{e(X)} implies
that with the minimal Stinespring representation of the channel
$\mathcal{I}(\Omega,\cdot)$ there exists a POVM $\{\mathsf{A}(x)\}$
satisfying $\mathcal{I}(X, T) =V^\ast(T \otimes A(X) )V$ for any $X$.
The effects $E= \mathsf{A}(X_1)$ and $F=\mathsf{A}(X_2)$ are
compatible and we obtain the wanted equations.

\end{proof}

Similar characterizations of compatibility between other combinations are easily derived. 
For instance, we have the following. 
(Because the proof is similar to the above proposition, 
it is omitted.)  

\begin{proposition}(operation-effect compatiblity)
Let $\Phi\hh$ be an operation and $E$ an effect. 
They are compatible if and only if there exist a Hilbert space $\hik'$, an isometry $V:\mathcal{H}
\to \hik \otimes \hik'$, and compatible effects $F,G \in \mathcal{L(K')}$ 
satisfying 
\begin{align*}
 \Phi\hh(T) & =V^*(T
\otimes F)V \qquad \forall T\in\mathcal{L(K)}, \\
 E & =V^*(\id \otimes G)V.
\end{align*}
\end{proposition} 
%

\subsection{Incompatibility in terms of Kraus operators}\label{sec:kraus}

In this subsection we will focus on the situation when the input and output Hilbert spaces are the same, $\hi$.
The Kraus decomposition theorem \cite{SEO83} states that a map $\Phi\ss:\trh\to\trh$ is an operation if and only if there exists a countable set of bounded operators $\{K_j\}_{j\in J}\subset\lh$, labeled by an index set $J$, such that 
\begin{equation}\label{eq:kraus}
\Phi\ss(\cdot) = \sum_{j\in J} K_j \cdot K_j^\ast \, , \qquad \sum_{j\in J} K_j^* K_j\leq \id \, .
\end{equation}
For a fixed operation $\Phi$, the choice of operators $K_j$, referred to as \emph{Kraus operators}, is not unique.
In any case, when comparing two Kraus decompositions we can always assume that they 
have the same number of elements by adding null operators if necessary. 
We typically choose $J\subseteq\nat$.

Suppose $\I$ is an instrument. 
We fix a Kraus decomposition $\{K_{x;j}\}$ for each operation $\I(x,\cdot)$, hence $\I$ can be written in the form
\begin{equation*}
\I\ss(x,\varrho) = \sum_j K_{x;j} \varrho K_{x;j}^\ast \qquad \forall \varrho\in\sh \, .
\end{equation*}
Conversely, a countable set of bounded operators $\{K_j\}_{j\in J}\subset\lh$ that satisfies $\sum_j K_j^\ast K_j = \id$ determines an instrument.
We can simply choose $\Omega=J$ and define $\I\ss(j,\varrho) = K_{j} \varrho K_{j}^\ast$.

Since instruments can be written in Kraus decomposition, it is clear that the relations of compatibility and weak compatibility can be formulated in terms of Kraus operators.
In the following we give formulations for the operation-operation and operation-effect pairs.

\begin{proposition}\label{prop:comp-kraus-oo}
Two operations $\Phi_1$ and $\Phi_2$ are:
\begin{itemize}
\item[(a)]  compatible if and only if there exists a sequence of bounded operators $\{K_j\}_{j\in J}$ and index subsets $J_1,J_2\subseteq J$ such that
\begin{equation}\label{eq:kraus-c-1}
\Phi_{1}\ss(\cdot) = \sum_{j\in J_1} K_j \cdot K^\ast_j \, , \quad 
\Phi_{2}\ss(\cdot)= \sum_{j\in J_2} K_j \cdot K_j^\ast \, 
\end{equation}
and
\begin{equation}
\sum_{j\in J} K^\ast_j K_j  = \id \, ;
\end{equation}
\item[(b)]  weakly compatible if and only if there exist sequences of bounded operators $\{K_j\}_{j\in J}$, $\{L_j\}_{j\in J}$ and index subsets $J_1,J_2\subseteq J$  such that
\begin{equation}\label{eq:opopkrausdecomposition}
\Phi_{1}\ss(\cdot) = \sum_{j\in J_1} K_j \cdot K^\ast_j \, , \quad 
\Phi_{2}\ss(\cdot)= \sum_{j\in J_2} L_j \cdot L_j^\ast 
\end{equation}
and
\begin{eqnarray}
\sum_{j\in J} K^\ast_j K_j  &=& \sum_{j\in J} L_j^\ast L_j=\id\, ,\\
\sum_{j\in J} K_j\cdot K^\ast_j  &=& \sum_{j\in J} L_j\cdot L^\ast_j\, .
\end{eqnarray}
\end{itemize}
\end{proposition}

\begin{proof}
\begin{itemize}
\item[(a)] See \cite{HeReStZi09}, Prop.~2.
\item[(b)] The ``if'' part is simple --- define $\I_{1}\ss(j,\varrho):=K_j\varrho K_j^\ast$ 
and $\I_{2}\ss(j,\varrho):=L_j\varrho L_j^\ast$. 
Then clearly $\Phi_1$ is part of $\I_1$ and $\Phi_2$ is part of $\I_2$ while equality $\I_1(J,\cdot)=\I_2(J,\cdot)$ holds.

The ``only if'' part is proved as follows. 
Suppose $\Phi_1$ and $\Phi_2$ are weakly compatible.
Then there exist instruments $\I_1$ and $\I_2$ such that $\I_1(\Omega_1,\cdot)=\I_2(\Omega_2,\cdot)$ while $\Phi_a(\cdot)=\I_a(X_a,\cdot)$ for $a=1,2$ and some $X_1$ and $X_2$.
Taking union of Kraus decompositions for $\Phi_1$ and $\I_1(\Omega_1\setminus X_1,\cdot)$ we obtain Kraus decomposition $\{K_j\}_{j\in J}$ of $\I_1(\Omega_1,\cdot)$ such that $\Phi_1$ is expressed via the subset $J_1\subseteq J\subseteq\nat$ of these Kraus operators.
Similarly we obtain Kraus operators for the second instrument  $\{L_j\}_{j\in J'}$ such that $\Phi_2$ is decomposed via subset $J_2\subseteq J'\subseteq\nat$ of these Kraus operators.
The index set can be chosen to be $\nat$ for both decompositions, as we can always supplement a set of Kraus operators by zero operators.
Thus, Eq.~(\ref{eq:opopkrausdecomposition}) follows.
The remaining two equations follow from the fact that $\I_1\hh(\Omega_1,\id)=\I_2\hh(\Omega_2,\id)=\id$ and that $\I_1(\Omega_1,\cdot)=\I_2(\Omega_2,\cdot)$.
\end{itemize}
\end{proof}

In a similar way we can also prove the following result for operation-effect pairs.

\begin{proposition}\label{prop:comp-kraus-oe}
An operation $\Phi$ and an effect $E$ are:
\begin{itemize}
\item[(a)] compatible if and only if there exists a sequence of bounded operators $\{K_j\}_{j\in J}$ and index subsets $J_1,J_2\subseteq J$ such that
\begin{equation}\label{eq:kraus-oe-1}
\Phi\ss(\cdot) = \sum_{j\in J_1} K_j \cdot K^\ast_j \, , \quad 
E= \sum_{j\in J_2} K_j^\ast K_j \, 
\end{equation}
and
\begin{equation}
\sum_{j\in J} K^\ast_j K_j  = \id \, .
\end{equation}

 \item[(b)] weakly compatible if and only if there exist sequences of bounded operators $\{K_j\}_{j\in J}$, $\{L_j\}_{j\in J}$ and index subsets $J_1,J_2\subseteq J$  such that
\begin{equation}
\Phi\ss(\cdot) = \sum_{j\in J_1} K_j \cdot K^\ast_j \, , \quad E= \sum_{j\in J_2} L_j^\ast L_j 
\end{equation}
and
\begin{subequations}
\label{eq:kraussums}
\begin{eqnarray}
\sum_{j\in J} K^\ast_j K_j  &=& \sum_{j\in J} L_j^\ast L_j=\id\, ,\\
\sum_{j\in J} K_j\cdot K^\ast_j  &=& \sum_{j\in J} L_j\cdot L^\ast_j\, .
\end{eqnarray}
\end{subequations}
 \end{itemize}
\end{proposition}

\begin{proof}
The proofs of these claims follow from Prop.~\ref{prop:comp-kraus-oo} when taken together with Prop.~\ref{prop:simple-2} and Prop.~\ref{prop:simple-3}b, respectively.
\end{proof}


\section{Incompatibility in terms of measurement models}\label{sec:memo}

The concepts and results introduced so far can be put into a wider perspective by considering measurement models.
We start by recalling some basic definitions from quantum measurement theory \cite{QTM96}. 

\begin{figure}
\begin{center}
\includegraphics[width=6.0cm]{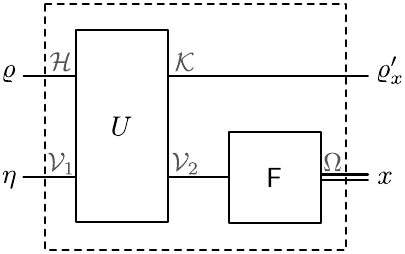}
\end{center}
\caption{\label{fig:memo} Measurement model $\memo=(\hv_1,\hv_2,\eta,U,\F)$. 
\newline\noindent
Here $\eta$ is an initial state on the ancillary system $\hv_1$ which is connected with the measured system $\hi$ by a global unitary operator $U$. 
}
\end{figure}

\begin{definition}
A quintuple $\memo=(\hv_1,\hv_2,\eta,U,\F)$ is a (generalized) measurement model if 
\begin{tabbing}
\hspace{1cm} \= $\hv_1$, $\hv_2$ \=  are Hilbert spaces, \\
\> $\eta$ \> is a state on $\hv_1$,\\
\> $U$ \> is a unitary operator from $\hi\otimes\hv_1$ 
to $\hik\otimes\hv_2$, \\
\> $\F$ \>  is an observable in $\hv_2$. 
\end{tabbing}
\end{definition}

The observable $\F$, called \emph{pointer observable}, gives us a measurement outcome $x\in \Omega$.
An input state $\varrho$ is transformed into a state $\varrho'_x$ (conditioned on $x$) --- see Fig.~\ref{fig:memo}.
The measurement outcome probabilities and the state transformations are given by the usual quantum formulae.
Namely, an outcome $x\in\Omega$ is recorded with the probability
\begin{equation}
\label{eq:memo_probability}
p(x\mid\varrho)=\tr{U\varrho\otimes\eta U^\ast \id_\hik\otimes\F(x)}
\end{equation}
and the input state $\varrho$ transforms into the unnormalized state $\varrho'_x$,
\begin{equation}
\label{eq:memo_state}
\varrho'_x = \trvout{U\varrho\otimes\eta U^\ast \id_\hik\otimes\F(x)} \, .
\end{equation}
Here $\trvout{\cdot}$ denotes the partial trace over ancillary Hilbert space $\hv_2$. 
In these formulas we considered only simple events which are of the form ``The obtained measurement outcome is $42$.''
We do not have to consider only those events that correspond to single measurement outcomes $x\in\Omega$, but we can group measurement outcomes into subsets. This means that we can also consider events of the form ``The obtained measurement outcome is between $1$ and $10$.''
Therefore we can replace $x$ by $X$ in Eqs.~(\ref{eq:memo_probability}) and (\ref{eq:memo_state}).
Similarly as the definitions of ``being part of an instrument'' we can define useful notions of ``being part of a measurement model'' --- we say that:
\begin{itemize}
\item an effect $E$ is part of $\memo$ if there exists a set $X\subseteq\Omega$ such that
\begin{align}
E =  \trvin{\id_\hi\otimes\eta U^\ast \id_\hik\otimes\F(X)U} \, ;
\end{align}
\item an operation $\Phi$ is part of $\memo$ if there exists a set $X\subseteq\Omega$ such that
\begin{equation}
\Phi\ss(\varrho) = \trvout{U\varrho\otimes\eta U^\ast \id_\hik\otimes\F(X)}  \qquad \forall\varrho\in\sh  ,
\end{equation}
or equivalently
\begin{equation}
\Phi\hh(T) = \trvin{\id_\hi\otimes\eta U^\ast T \otimes\F(X) U }   \qquad \forall T \in\lh \, .
\end{equation}
\end{itemize}

Being part of $\memo$ simply means that, having the measurement model $\memo$ available, we can implement $E$ (resp. $\Phi$) by ignoring everything else but some component of $\memo$ --- see Fig.~\ref{fig:memo_effect_operation}.
Since channels are special types of operations, we see that:
\begin{itemize}
\item a channel $\Lambda$ is part of $\memo$ if
\begin{equation}
\label{eq:channel}
\Lambda\ss(\varrho)= \trvout{U\varrho\otimes\eta U^\ast} \qquad \forall\varrho\in\sh \, ,
\end{equation}
or equivalently
\begin{equation}
\label{eq:channelH}
\Lambda\hh(T) = \trvin{\id_\hi\otimes\eta U^\ast T \otimes\id_{\hv_2} U }  \qquad \forall T \in\lh \, .
\end{equation}
\end{itemize}
Clearly, each measurement model determines a unique channel. 
Useful observation is that this channel does not depend on the choice of the pointer observable $\F$, but only on the ancillary state $\eta$ and measurement coupling $U$.

\begin{figure}
\begin{center}
a) \includegraphics[width=4.0cm]{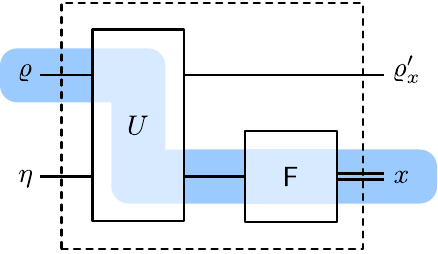}\hskip1cm b) \includegraphics[width=4.0cm]{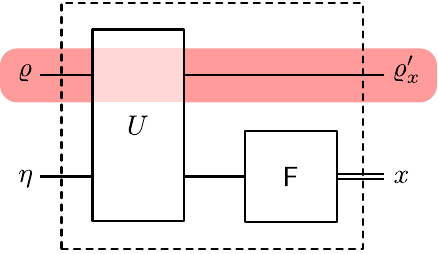}
\end{center}
\caption{\label{fig:memo_effect_operation} Some parts of a measurement model are  a) effect and b) operation.}
\end{figure}

As it was the case with the definitions through instruments, also here we need to take into consideration that observables and instruments do not describe single events but collections of possible events.
Since the measurement outcomes can be regrouped and relabeled after the measurement is performed, we again include a \emph{pointer function} into our description.
 Thus, we say that
\begin{itemize}
\item an observable $\A$ with an outcome set $\Omega'$ is part of $\memo$ if there exists a function $f:\Omega\to\Omega'$ such that
\begin{equation}\label{eq:observable}
\A(X) =  \trvin{\id_\hi\otimes\eta U^\ast \id_\hik\otimes\F(f^{-1}(X))U} \qquad \forall X\subseteq\Omega'\, ;
\end{equation}
\item an instrument $\I$  with an outcome set $\Omega'$ is part of $\memo$ if there exists a function $f:\Omega\to\Omega'$ such that $\forall\varrho\in\sh \, , X\subseteq\Omega'$
\begin{equation}\label{eq:instmemoS}
\I\ss(X,\varrho) = \trvout{U\varrho\otimes\eta U^\ast \id_\hik\otimes\F(f^{-1}(X))} \, ,
\end{equation}
or equivalently $\forall T\in\lh \, , X\subseteq\Omega'$
\begin{equation}\label{eq:instmemoH}
\I\hh(X,T) = \trvin{\id_\hi\otimes\eta U^\ast T\otimes\F(f^{-1}(X))U}  \, .
\end{equation}
\end{itemize}

It has been proved in \cite{Ozawa84} that every instrument $\I\hh$ from $\mathcal L(\hi)$ to $\mathcal L(\hi)$ is part of a measurement model.
We are going to need not only that result but also its proof in the following text, so we present the proof of this fact in the case of a finite outcome space for reader's convenience.

\begin{proposition}\label{prop:ItoMEMO}
Let $\I\hh$ be an instrument from 
$\mathcal L(\hik)$ to $\mathcal L(\hi)$.  Then there exists a measurement model 
$\memo=(\hv_1,\hv_2,\eta,U,F)$ satisfying 
\begin{equation}
\I\hh(X,T)= \trvin{\id_\hi\otimes\eta U^\ast T \otimes\F(X)U}  \, 
\end{equation}
for all subsets $X\subseteq \Omega$. 
\end{proposition}

\begin{figure}
\begin{center}
\includegraphics{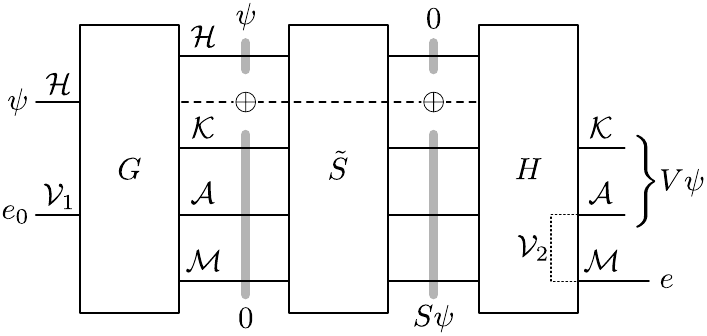}
\end{center}
\caption{\label{fig:ItoMEMO} Stinespring representation with extension showing that for every instrument there exists some measurement model.}
\end{figure}

\begin{proof}
We fix a Stinespring representation ($\mathcal{A}$, $V$) for the channel $\I\hh(\Omega,\cdot)$, i.e.~$\mathcal{A}$ is a Hilbert space and $V : \hi \to \hik \otimes \mathcal{A}$ is 
an isometry. 
By Prop.~\ref{e(X)} there exists an observable $\E$ on $\mathcal{L(A)}$ satisfying 
\begin{equation}
\I\hh(X,T)=V^*(T\otimes \E(X))V 
\end{equation}
for all $T\in\mathcal{L(K)}$ and $X\subseteq\Omega$. 
We introduce (see also Fig.~\ref{fig:ItoMEMO}) an auxiliary Hilbert space $\mathcal{M}$ whose dimension is infinite and a unit vector $e\in \mathcal{M}$ to define an isometry $S:\hi \to \hik \otimes \mathcal{A}\otimes \mathcal{M}$ by 
\begin{equation}
S\psi:=(V\psi) \otimes e \, .
\end{equation}
Then
\begin{equation}
\I\hh(X,T)=S^*(T\otimes \E(X)\otimes \id_{\mathcal{M}})S \, .
\end{equation}
The isometry $S$ can be dilated to a unitary operator $\widetilde{S}$ from $\hi \oplus (\hik\otimes \mathcal{A}\otimes \mathcal{M})$ to itself by setting
\begin{equation}
\widetilde{S}(\psi\oplus \phi):=-S^*\phi\oplus (\sqrt{\id -S S^*}\phi +S\psi).
\end{equation}
Clearly, $\widetilde{S}(\psi\oplus 0)=0\oplus S\psi$. 
We introduce another infinite-dimensional Hilbert space $\hv_1$ and a unit vector $e_0\in \hv_1$.  
We can define a unitary operator $G: \hi \otimes \hv_1 \to \hi\oplus (\hik\otimes \mathcal{A}\otimes \mathcal{M})$ satisfying 
\begin{equation}
G\psi \otimes e_0=\psi \oplus 0.
\end{equation}  
As $\mathcal{M}$ is infinite-dimensional, we can also define a unitary operator $H: \hi\oplus (\hik\otimes \mathcal{A}\otimes \mathcal{M}) \to \hik\otimes \mathcal{A} \otimes \mathcal{M}$ 
satisfying 
\begin{equation}
H(0\oplus \phi\otimes e)=\phi\otimes e
\end{equation}
for all $\phi\in \hik\otimes \mathcal{A}$. 
We define $\hv_2:=\mathcal{A}\otimes \mathcal{M}$. 
Then $U:=H\widetilde{S}G: \hi \otimes \hv_1 \to\hik\otimes \hv_2$ is a unitary operator and it satisfies
\begin{equation}
U (\psi\otimes e_0) =S\psi=(V\psi) \otimes e.
\end{equation}
Set $\eta:=\kb{e_0}{e_0}$ and $\F(X):=\E(X)\otimes \id$. 
Now, for all $\psi\in\hi$, we have
\begin{eqnarray*}
\ip{\psi}{\I\hh(X,T)\psi}
&=&\ip{V\psi}{T\otimes \E(X))V\psi}\\
&=&\ip{(V\psi)\otimes e}{T\otimes \E(X)\otimes \id V(\psi)\otimes e}\\
&=&\ip{\psi\otimes e_0}{U^*(T\otimes \F(X))U \psi\otimes e_0}\\
&=&\langle \psi | \mbox{tr}_{\hv_1}[\id_{\mathcal{H}}\otimes
\eta U^* T\otimes \F(X) U]  \psi\rangle 
%
\end{eqnarray*}
This concludes the proof. 
\end{proof}

Based on this result we can now prove the following. 
\begin{proposition}\label{prop:memo_almost_compat}
Let $\I_1\hh$ and $\I_2\hh$ be two instruments from $\mathcal L(\hik)$ to $\mathcal L(\hi)$. If they satisfy $\I_1\hh(\Omega_1,\cdot)=\I_2\hh(\Omega_2,\cdot)$, then it is possible to set their measurement models as $\memo_1=(\hv_1,\hv_2,\eta,U,\F_1)$ and $\memo_2=(\hv_1,\hv_2,\eta,U,\F_2)$. 
\end{proposition}

\begin{proof}
The construction of $\hv_1$, $\hv_2$, $\eta$ and $U$ of the measurement model in the proof of Prop.~\ref{prop:ItoMEMO} depends only on $\I\hh(\Omega,\cdot)$ and is not observable $\E$ dependent. The only difference is possible in pointer observable which is $X$-dependent.
\end{proof}

\begin{proposition}\label{prop:memo_compat}
Two devices $\dev_1$ and $\dev_2$ are
\begin{itemize}
\item[(a)] compatible if and only if there exists a measurement model $\memo$ such that both $\dev_1$ and $\dev_2$ are parts of $\memo$;
\item[(b)] weakly compatible if and only if there exist two measurement models $\memo_1=(\hv_1,\hv_2,\eta,U,\F_1)$ and $\memo_2=(\hv_1,\hv_2,\eta,U,\F_2)$, differing only in their pointer observables $\F_1$ and $\F_2$, such that $\dev_1$ is part of $\memo_1$ and $\dev_2$ is part of $\memo_2$.
\end{itemize}
\end{proposition}

\begin{proof}
Propositions \ref{prop:ItoMEMO} and \ref{prop:memo_almost_compat} prove the only if part of both statements. The if part is easily concluded from Eqs.~(\ref{eq:instmemoS}) and (\ref{eq:instmemoH}) which show, that if device $\dev$ is part of measurement model $\memo$, then it is also part of instrument $\I$ corresponding to $\memo$ as all the possible devices that are parts of $\memo$ can be recovered also from $\I$.
\end{proof}

We can thus see that compatibility is equivalent to the existence of a common measurement model, while weak compatibility is equivalent to the existence of a common measurement model up to different choices of pointer observables.
This equivalence reveals the clear operational meaning behind these concepts.
We can even use Prop.~\ref{prop:memo_compat} as an alternative route to prove facts about compatibility and weak compatibility --- this is demonstrated in the following example that proves Prop.~\ref{prop:elementary} using measurement models.

\begin{example}\label{ex:swap} 
Let us consider a measurement model $\memo=(\hi,\hi,\eta,U,\F)$, where $\eta$ is an arbitrary state and $U=\swap$ is the \emph{swap operator} defined as
\begin{equation}
\swap \psi\otimes\varphi = \varphi\otimes\psi \qquad \forall\psi,\varphi\in\hi \, .
\end{equation}
From Eq.~\eqref{eq:observable} we get $\A=\F$ (when $f$ is chosen to be the identity function), which means that every observable is part of the same measurement model up to a change of a pointer observable (see also Fig.~\ref{fig:swap}).
Thus, we obtain an alternative proof of Prop.~\ref{prop:elementary}b.

\begin{figure}
\begin{center}
\includegraphics[height=3.5cm]{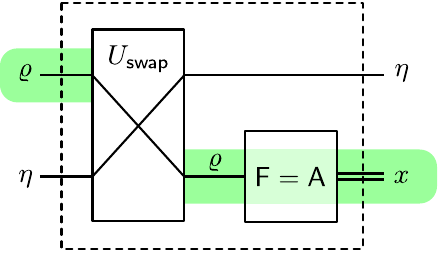}
\end{center}
\caption{\label{fig:swap} A measurement model with $U=\swap$ shows that all observables are weakly compatible as the choice of pointer observable $\F$ corresponds exactly to the observable in consideration $\A$.}
\end{figure}

Note that we can also see that Prop.~\ref{prop:elementary}a holds since Eq.~(\ref{eq:channel}) and (\ref{eq:channelH}) show that changing the pointer observable does not change the channel. This in turn means that a channel is incompatible with some device if and only if they are strongly incompatible, as you can set the pointer observable for channel to be the same as the pointer observable for the device.

\end{example}

\section{Examples of possible relations}\label{sec:examples}

In this section we show that all the incompatibility relations are possible between operations and effects, except the strong incompatibility of two effects.
The latter was noticed to be impossible in Proposition \ref{prop:elementary}. 
The overall situation is summarized in Table \ref{table:1}.
The first row in Table \ref{table:1} is clear --- there are compatible devices in all three different pairs.
The last entry of the second row is also clear since there exist incompatible effects.
We will demonstrate that the remaining four situations (circled) are possible.

\begin{table}[htdp]
\caption{Summary of possible relations. The circled points are demonstrated in this section.}
\begin{center}
\begin{tabular}{lccc}
\hline
& op-op & op-ef & ef-ef \\ \hline\hline
 compatible & \checkmark & \checkmark & \checkmark \\
 incompatible but weakly compatible & \textcircled{\checkmark} &  \textcircled{\checkmark} & \checkmark\\
strongly incompatible &  \textcircled{\checkmark} &  \textcircled{\checkmark} & $\times$\\
\hline
\end{tabular}
\end{center}
\label{table:1}
\end{table}

Our examples are all related to qubit systems, hence $\hi=\hik=\complex^2$.
Let $\sigma_x$, $\sigma_y$ and $\sigma_z$ be the Pauli operators on $\complex^2$.
We denote $P_{j} = \half (\id + \sigma_j)$ and $P_{-j} = \half (\id - \sigma_j)$ for $j=x,y,z$, and $P_{j}$ and $P_{-j}$ are hence one-dimensional projections. 

\begin{example}\label{ex:onlyweaklycompatibleops}
(\emph{Two operations that are incompatible but not strongly incompatible})
We consider operations $\Phi_{1}\ss(\varrho)=P_x\varrho P_x$ and $\Phi_{2}\ss(\varrho)=\half \sigma_x\varrho \sigma_x$.
They are both pure operations (i.e.~have only one Kraus operator), hence compatible if and only if they are comparable or $\Phi_1+\Phi_2$ is an operation (see the discussion before Prop.~\ref{prop:channel_for_wc}).
It is therefore easy to verify that they are incompatible.
To see that $\Phi_1$ and $\Phi_2$ are weakly compatible, we define a channel $\Lambda$ by
\begin{equation*}
\Lambda\ss(\varrho)=\half \varrho + \half \sigma_x\varrho \sigma_x \, .
\end{equation*}
Substituting $\id=P_{x} + P_{-x}$ and $\sigma_x = P_{x} - P_{-x}$ we see that $\Lambda$ can be written in the alternative form
\begin{equation*}
\Lambda\ss(\varrho)= P_{x} \varrho P_{x} + P_{-x}\varrho P_{-x} \, .
\end{equation*}
We thus have $\Phi_1\leq\Lambda$ and $\Phi_2\leq\Lambda$, therefore $\Phi_1$ and $\Phi_2$ are weakly compatible by Prop.~\ref{prop:channel_for_wc}.
\end{example}

It is easy to give examples of strongly incompatible channels as any pair of two different channels is incompatible.
In the following we provide more interesting example where the strongly incompatible operations are not channels.

\begin{example}\label{ex:stronglyincompatibleops}
(\emph{Two operations that are strongly incompatible})
We consider operations $\Phi_{1}\ss(\varrho)=P_{x}\varrho P_{x}$ and $\Phi_{2}\ss(\varrho)=P_{z}\varrho P_{z}$.
Since $\id-\Phi_1\hh(\id)=P_{-x}$ is a rank-1 operator, by Prop.~\ref{prop:order-rank-1} the operation $\Phi_1$ satisfies $\Phi_1\leq\Lambda_1$ for some channel $\Lambda_1$ iff
\begin{equation*}
\Lambda\ss_1(\varrho)=P_x \varrho P_x + \tr{\varrho P_{-x}} \xi_1
\end{equation*}
for some state $\xi_1$.
Similarly, the operation $\Phi_2$ satisfies $\Phi_2\leq\Lambda_2$ for a channel $\Lambda_2$ if and only if
\begin{equation*}
\Lambda\ss_2(\varrho)=P_z \varrho P_z + \tr{\varrho P_{-z}} \xi_2
\end{equation*}
for some state $\xi_2$.
We have $\Lambda\ss_1(P_x)=P_x$ and $\Lambda\ss_2(P_x)=\half P_z + \half \xi_2$. 
Since $P_x \neq \half P_z + \half \xi_2$ for any choice of $\xi_2$, we conclude that $\Lambda_1\neq\Lambda_2$ irrespective of the choices of $\xi_1$ and $\xi_2$.
Therefore, $\Phi_1$ and $\Phi_2$ are strongly incompatible.
\end{example}

\begin{example}
(\emph{Effect and operation that are incompatible but not strongly incompatible})
We consider the projection $P_x$ and the L\"uders operation $\Phi\ss(\varrho)=P_z \varrho P_z$.
Since the effects $P_x$ and $\Phi\hh(\id)=P_z$ are incompatible, we conclude from Prop.~\ref{prop:simple} that $P_x$ and $\Phi$ are incompatible.
To see that $P_x$ and $\Phi$ are weakly compatible, let us fix normalized eigenvectors $\phi_{x\pm}$ and $\phi_{z\pm}$ for $\sigma_x$ and $\sigma_z$, respectively. 
We observe that the channel $\varrho \mapsto \tr{\varrho} P_z$ can be written in the alternative forms
\begin{equation*}
 \tr{\varrho} P_z = \kb{\phi_{z+}}{\phi_{z+}} \varrho  \kb{\phi_{z+}}{\phi_{z+}} + \kb{\phi_{z+}}{\phi_{z-}} \varrho  \kb{\phi_{z-}}{\phi_{z+}}
\end{equation*}
and
\begin{equation*}
 \tr{\varrho} P_z = \kb{\phi_{z+}}{\phi_{x+}} \varrho  \kb{\phi_{x+}}{\phi_{z+}} + \kb{\phi_{z+}}{\phi_{x-}} \varrho  \kb{\phi_{x-}}{\phi_{z+}} \, .
\end{equation*}
Using Prop.~\ref{prop:comp-kraus-oe} we then conclude that $P_x$ and $\Phi$ are weakly compatible.
\end{example}

\begin{example}
(\emph{Effect and operation that are strongly incompatible})
We consider the projection $P_x$ and the L\"uders operation $\Phi\ss(\varrho)=A^{\half} \varrho A^{\half}$ with $A=P_z+ \half P_{-z}$.
Let us make a counter assumption that $P_x$ and $\Phi$ are weakly compatible.
By Prop.~\ref{prop:simple-3} this means that there is an operation $\Phi'$ weakly compatible with $\Phi$ and satisfying ${\Phi'}\hh(\id)=P_x$.
By Prop.~\ref{prop:channel_for_wc} there exists a channel $\Lambda$ such that $\Phi\leq\Lambda$ and $\Phi'\leq\Lambda$.
The effect $\id-A=\half P_{-z}$ is rank-1, hence by Prop.~\ref{prop:order-rank-1} we conclude that $\Phi\leq\Lambda$ is possible only if $\Lambda$ has the form
\begin{equation}\label{eq:lambda-1}
\Lambda\ss(\varrho)=\Phi\ss(\varrho) + \half \tr{\varrho P_{-z}} \xi 
\end{equation}
for some state $\xi$.
On the other hand, since ${\Phi'}\hh(\id)=P_x$ and $P_x$ is rank-1, then by Prop.~8 of \cite{HeWo10} we first have
\begin{equation*}
{\Phi'}\ss(\varrho)=\tr{\varrho P_x} \xi_1
\end{equation*}
for some state $\xi_1$ and, by applying Prop.~\ref{prop:order-rank-1} again, we conclude that $\Phi'\leq\Lambda$ is possible only if $\Lambda$ has the form
\begin{equation}\label{eq:lambda-2}
{\Lambda}\ss(\varrho)=\tr{\varrho P_{x}} \xi_1 +  \tr{\varrho P_{-x}} \xi_2  
\end{equation}
for some states $\xi_1,\xi_2$.
Inserting $\varrho=P_z,P_{-z}$ in both \eqref{eq:lambda-1} and \eqref{eq:lambda-2}, and equaling them, we obtain 
\begin{align*}
& P_z = \half \xi_1 + \half \xi_2\, ,
& \half P_{-z} + \half \xi = \half \xi_1 + \half \xi_2 \, .
\end{align*}
But $P_z\neq \half P_{-z} + \half \xi$ for any state $\xi$, hence we arrive to a contradiction and the counter assumption is therefore false.
\end{example}

\section{Conclusions}

The notions of coexistence and joint measurability are in this paper united into a single definition of compatibility. This is done by relating all measurement devices to instruments. This definition then allows one to study the compatibility of objects also of different types, e.g.~operations and effects. We defined also a tighter notion of incompatibility called strong incompatibility. 
These notions are explored by means of the Stinespring dilation, which shows an intriguing relation of compatibility features of the studied devices to the compatibility of the effects/observables underlying the construction of the dilation. These notions were also studied by Kraus decomposition. Relating the compatibility relations to measurement models illustrates an operational meaning of these notions in a simple way --- compatibility of two devices is conditioned by a single measurement model for both devices, while for weak compatibility the two devices are required to have a single measurement model up to the pointer observable. Both notions of compatibility are distinct in such a way that there exist devices which are weakly compatible, yet still incompatible.


\section{Acknowledgement}

T.H.~acknowledges financial support from the Academy of Finland (grant no. 138135).  
T.M.~acknowledges JSPS KAKENHI (grant no. 22740078).
D.R.~acknowledges financial support from the pro\-ject COQUIT.


\bibliographystyle{plain}

\end{document}